\documentclass{CSML}
\pdfoutput=1

\usepackage{lastpage}
\newcommand{\dOi}{14(2:4)2018}
\lmcsheading{}{1--\pageref{LastPage}}{}{}%
{Jun.~30, 2017}{Apr.~25, 2018}{}


\usepackage{hyperref}
\hypersetup{hidelinks}
\usepackage{amssymb}
\usepackage{amsmath}
\usepackage{amssymb}
\usepackage{amsthm}

\renewcommand{\leq}{\leqslant}
\renewcommand{\geq}{\geqslant}

\def\fin{\textrm{end}}

\def\comp{\textrm{same-comp}}

\def\bbN{\mathbb{N}}

\def\stra{\mathfrak{A}}

\def\cC{\mathcal{C}}
\def\cocC{\textsc{co-}\mathcal{C}}
\def\cD{\mathcal{D}}
\def\cF{\mathcal{F}}

\def\cP{\mathcal{P}}

\def\vz{{\bar{z}}}

\def\nram#1{\textsc{NRAM}[#1]}

\def\ntime#1{\textsc{NTIME}[#1]}
\def\dspace#1{\textsc{DSPACE}[#1]}
\def\cone{\textsc{co-NE}}
\def\ne{\textsc{NE}}

\def\np{\textsc{NP}}

\def\pspace{\textsc{PSPACE}}
\def\expspace{\textsc{EXPSPACE}}

\def\specarity#1{\textsc{Spec}\mathrm{(arity}\ #1\mathrm{)}}

\def\spec{\textsc{Spec}}

\def\fo{\textsf{FO}}

\def\OMIT #1{}

\begin{document}


\title[A note on first-order spectra with binary relations]{A note on first-order spectra with binary relations}

\author[E. Kopczy\'nski]{Eryk Kopczy\'nski}
\address{University of Warsaw}
\email{erykk@mimuw.edu.pl}

\author[T. Tan]{Tony Tan}
\address{National Taiwan University}
\email{tonytan@csie.ntu.edu.tw}

\begin{abstract}
The {\em spectrum} of a first-order sentence
is the set of the cardinalities of its finite models.
In this paper, we consider the spectra of sentences over binary relations
that use at least three variables.
We show that for every such sentence $\Phi$, there is a sentence $\Phi'$
that uses the same number of variables, but only one symmetric binary relation, such that
its spectrum is linearly proportional to the spectrum of $\Phi$.
Moreover, the models of $\Phi'$ are all bipartite graphs.
As a corollary, we obtain that to settle Asser's conjecture,
i.e., whether the class of spectra is closed under complement,
it is sufficient to consider only sentences using only three variables
whose models are restricted to undirected bipartite graphs.
\end{abstract}

\maketitle

\section{Introduction}
\label{sec:intro}

The notion of first-order spectrum was first defined by Scholz~\cite{Scholz52}.
Formally, the spectrum of a (first-order) sentence $\varphi$ (with the equality predicate), 
denoted by $\spec(\varphi)$,
is the set of cardinalities of finite models of $\varphi$.
A set is called a spectrum,
if it is the spectrum of a first-order sentence.
Let $\spec$ denote the class of all spectra.

One of the first and well known problems in finite model theory, called {\em Asser's conjecture}, 
asks whether the complement of a spectrum is also a spectrum~\cite{Asser55}.
It turns out to be equivalent to $\ne$ vs. $\cone$ 
problem~\cite{JS74,Fagin73,Fagin}.\footnote{$\ne$ is the class of languages accepted
by a non-deterministic (possibly multi-tape) Turing machine with run time
$O(2^{kn})$, for some constant $k > 0$.}
More specifically, it is shown that the class $\ne$ is captured precisely by $\spec$
in the following sense:
For every spectrum $A$, the language that consists of the binary representations of the numbers in $A$
belongs to the class $\ne$, and vice versa,
for every language $L \subseteq 1\cdot \{0,1\}^*$ , i.e., it consists of only words that start with symbol $1$,
if $L \in \ne$, then the set of integers whose binary representations are in $L$ is a spectrum. 
For a more comprehensive treatment on the spectrum problem
and its history,
we refer interested readers to an excellent survey by 
Durand, Jones, Makowsky and More~\cite{DJMM12}, and the references therein.

It is reasonable to say that a definitive solution of Asser's conjecture seems still far away.
Thus, it is natural to consider the spectra of some restricted classes of first-order logic.
Fagin~\cite{Fagin75} was the first to notice that to settle Asser's conjecture,
it is sufficient to consider only first-order logic over graphs.
More formally, he showed that 
for every spectrum $A$, there is a positive integer $k>0$ such that 
$\{n^k \mid n \in A\}$ is the spectrum of a sentence using only one binary relation symbol.
Implicitly, it implies that if there is a spectrum whose complement is not a spectrum,
then there is such a spectrum of first-order sentence using only one binary relation~\cite{Fagin,Fagin73},
i.e., Asser's conjecture can be reduced to first-order sentences over graphs.

Durand and Ranaivoson~\cite{DurandR96} considered the class of spectra
of sentences using only unary function symbols and
proved that it is included in
the class of spectra of sentences using only one binary relation.
In particular, they established that the spectra of sentences using only unary function symbols
are exactly the spectra of sentences using one binary relation when 
the models for the latter are restricted to directed graphs of bounded outdegree.
They also showed that there is a sentence $\varphi$ using {\em two} unary functions such that 
the language $\{1^n \mid n \in \spec(\varphi)\}$ is $\np$-complete.
That two unary functions are necessary to obtain an $\np$-complete language
is shown immediately by Durand, Fagin and Loescher~\cite{DurandR96,DurandFL97},
where they show that the spectrum of a first-order sentence using only {\em one} unary function symbol
is a semilinear set.

Complementing Fagin's result,
we showed that Asser's conjecture can be reduced to
sentences using only three variables and multiple binary relations~\cite{KT15tocl}.
The three variable requirement seems to be optimal,
as we also showed that the class of the spectra of sentences using two variables
and counting quantifiers is precisely the class of semilinear sets
and closed under complement~\cite{KT15sicomp}.
In fact, we essentially showed that models of two-variable logic with counting
are simply collections of regular bipartite graphs.

In this paper we present the following result.

\begin{thm}
\label{theo:main}
For every sentence $\Phi$ using at least three variables
over binary relation symbols $R_1,\ldots,R_m$,
there is a sentence $\Phi'$ over a single binary relation symbol $E$ that
uses the same number of variables as $\Phi$ such that:
\begin{eqnarray*}
\spec(\Phi') & = &  \{pn+q \mid n \in \spec(\Phi)\},\qquad\qquad\textrm{for some integers}\ p \ \textrm{and}\ q.
\end{eqnarray*}
Moreover, every model of $\Phi'$ is an undirected bipartite graph.
\end{thm}

Since addition, subtraction, multiplication and division by constants
can be computed in linear time (in the length of the binary representation of the input number), 
the spectra of $\Phi$ and $\Phi'$ do not differ complexity-wise.
Combined with our earlier result~\cite[Corollary~3.5]{KT15tocl} that Asser's conjecture can be reduced
three variable sentences with binary relations,
Theorem~\ref{theo:main} immediately implies that
Asser's conjecture can be further reduced to three variable sentences using only 
one binary relation with models being restricted to bipartite graphs.
It is stated formally as Corollary~\ref{cor:asser} below.

\begin{cor}
\label{cor:asser}
The following two sentences are equivalent.
\begin{itemize}\itemsep=0pt
\item 
The class of first-order spectra is closed under complement.
\item
The complement of every spectrum of first-order sentence using only three variables whose models are all 
undirected bipartite graphs
is also a spectrum.
\end{itemize}
\end{cor}

Note that Corollary~\ref{cor:asser} strengthens the result by Fagin~\cite{Fagin75} 
which states that Asser's conjecture can be reduced to sentences (with arbitrary number of variables)
over graphs.
We also note the difference between Theorem~\ref{theo:main} and the result by Durand and Ranaivoson~\cite{DurandR96}
mentioned above.
In~\cite{DurandR96}, multiple {\em unary functions} are encoded using only one binary relation
(with the graphs being restricted to those with bounded outdegree),
whereas in Theorem~\ref{theo:main}, multiple {\em binary relations}
are encoded with one binary relation (albeit with linear blowup in the size of the model).

At this point, it is natural to ask whether every spectrum is 
the spectrum of a sentence over graphs, i.e., a sentence using only one relation symbol of arity $2$.
It turns out that a positive answer to this question will imply
the separation of a long standing open problem: $\ne\subsetneq \expspace$,
and thus, $\np\subsetneq \pspace$, as stated formally in Remark~\ref{remark} below.

\begin{rem}
\label{remark}
Let $\specarity {k}$ denote the class of spectra
of sentences using only relational symbols of arity $k$.
We will prove the following: {\em If $\spec=\specarity{k}$, for some integer $k$,
then $\ne\subsetneq \expspace$, and hence, $\np\subsetneq \pspace$.}

First, we show that $\specarity{k} \subseteq \dspace{2^{kn}}$,
where the input integer is written in binary form.
Let $\varphi$ be an $\fo$ sentence using relations of arity at most $k$.
To show that $\spec(\varphi) \in \dspace {2^{kn}}$, 
let $w$ be the input word that represents integer $N$ in binary form.
Each relation $R$ of arity $k$ with domain $\{1,\ldots,N\}$ takes $N^k = O(2^{k|w|})$ space.
So, each model $\stra$ with relations of arity at most $k$ takes $O(2^{k|w|})$ space.
Checking whether $\stra$ satisfies $\varphi$ takes additional $O(|w|)$ space.
To check whether $\varphi$ has a model of cardinality $N$,
one can simply check one by one every possible model with domain $\{1,\ldots,N\}$, 
each of which takes $O(2^{k|w|})$ space.
Therefore, $\spec(\varphi)\in \dspace{2^{kn}}$.

Now, by the space hierarchy theorem~\cite{SHL65},
$\dspace{2^{kn}} \subsetneq \expspace$.
Thus, if $\spec = \specarity{k}$, for some $k$, then $\ne \subsetneq \expspace$,
and by standard padding argument, it implies $\np\subsetneq\pspace$.
\end{rem}

\subsection*{Related work.}

It is already noted before that first-order logic over arbitrary vocabulary
is too vast a logic to work on.
A lot of work has been done to classify spectra based on the vocabulary,
notably on the arity of the relation and function symbols.
We will mention some of them here.
Interested readers can consult the cited papers and the references therein.

Let $\ntime {N^k}$ denote the class of sets of positive integers
(written in unary form) accepted by non-deterministic multi-tape
Turing machine in time $O(N^k)$, where $N$ is the input integer.
Lynch~\cite{Lynch82} showed that $\ntime {N^k} \subseteq \specarity {k}$, for every $k\geq 2$.
When $k=1$, the addition operator is required, i.e., $\ntime {N} \subseteq \specarity {1,+}$.
The converse of Lynch's theorem is still open.

Grandjean, Olive and Pudl\'ak established the variable hierarchy
for spectra of sentences using relation and function 
symbols~\cite{Grandjean84,Grandjean85,Grandjean90,GrandjeanO04,Pudlak75}.
Let $\nram {N^k}$ denote the class of sets of positive integers
accepted by a non-deterministic RAM in time $O(N^k)$, 
and $N$ is the input integer.
In his series of papers, Grandjean showed that
the class $\nram {N^k}$ is precisely the class of the spectra 
of first-order sentences written in prenex normal form
using only universal quantifiers and $k$ variables
with vocabulary consisting of relation and function symbols of 
arity $k$~\cite{Grandjean84,Grandjean85,Grandjean90}.
By Skolemisation, this result leads to the fact that for every integer $k \geq 1$,
the class of spectra of first-order sentences using relation and 
function symbols and $k$ variables
is precisely $\nram {N^k}$.
See also~\cite[Theorem~3.1]{GrandjeanO04}.

Grandjean~\cite{Grandjean90} also showed that the class $\nram {N}$
is precisely the class of spectra of sentences of the form $\forall x \varphi$,
where $\varphi$ is quantifier free and uses only unary functions.
Note that to express that a relation is a function requires three variables.
Since composition of functions can also be expressed with three (reusable) variables,
it implies that $\nram {N}$ is a subclass of the class of spectra involving only binary relations
and three variables.
By padding argument, it also implies that if Asser's conjecture is negative,
it suffices to consider only three-variable sentences using only binary relations.
This is similar to our result in~\cite{KT15tocl}.

A result similar to Theorem~\ref{theo:main} was also obtained by Durand, et. al.~\cite{DurandFL97}
where they showed that if $S$ is a spectrum involving $k$ unary functions,
then the set $\{kn \mid n \in S\}$ is a spectrum involving only {\em two} unary functions.
There is a strong evidence that the linear blow-up is unavoidable~\cite[Proposition~5.1]{DurandFL97}.
Durand and Ranaivoson~\cite{DurandR96} also showed that
every spectrum can be transformed (with polynomial blowup)
to a spectrum involving only unary functions, i.e.,
if $S$ is a spectrum involving $k$-ary functions, 
then $\{n^k \mid n \in S\}$ is a spectrum involving only unary functions.
Durand's thesis~\cite{Durand96} is rich with results in this direction.

Recently we also showed that there is a strict hierarchy of spectra
based on the number of variables used.
That is, more variables
yield larger class of spectra~\cite{KT15tocl}
when the vocabulary is restricted to relational symbols.

\subsection*{Organization.}
In the next section we will present the proof of Theorem~\ref{theo:main},
and we conclude with some remarks in Section~\ref{sec:concl}.

\section{Proof of Theorem~\ref{theo:main}}

In this paper, by graph we always mean undirected graph.
For a graph $G=(V,E)$ and a subset $V'\subseteq V$,
we denote by $G[V']$ the subgraph of $G$ induced by the subset $V'$.

Let $R_1,\ldots,R_m$ be binary relation symbols.
For $k\geq 0$, we denote by $\fo^{k}[R_1,\ldots,R_m]$
the class of $\fo$ formulas using $k$ variables
and binary relation symbols $R_1,\ldots,R_m$.
A formula is a sentence, if it has no free variable.
A formula is always written as $\varphi(z_1,\ldots,z_l)$
to indicate that $z_1,\ldots,z_l$ are the free variables in $\varphi$. 

An interpretation is written in a standard way $\stra = (A,R^{\stra}_1,\ldots,R^{\stra}_m)$,
where $A$ is a finite domain and each $R^{\stra}_i \subseteq A\times A$, for each $i=1,\ldots,m$.
As usual, $\stra\models \varphi$ denotes that the sentence $\varphi$ holds in $\stra$.
For a formula $\varphi(z_1,\ldots,z_l)$,
and for $i_1,\ldots,i_l \in A$, we write that $\varphi(i_1,\ldots,i_l)$ holds in $\stra$,
if $\varphi(z_1,\ldots,z_l)$ holds in $\stra$ by substituting each $z_j$ with $i_j$,
for every $j=1,\ldots,l$.

We reserve the symbol $E$ to be a binary relation symbol
that we insist to be always interpreted by a symmetric relation.
In the same way, we let $\fo^k[E]$ to be
the class of $\fo$ formulas using $k$ variables
and relation symbol $E$.
All models of sentences from $\fo^{k}[E]$ are graphs, so
we will use the standard notation $G=(V,E)\models \varphi$, or simply $G\models \varphi$,
to denote that $\varphi$ holds in $G$.

The following Lemma~\ref{lem:preserve-k-var} immediately implies Theorem~\ref{theo:main}.

\begin{lem}
\label{lem:preserve-k-var}
Let $k\geq 3$.
For every $\Phi \in\fo^k[R_1,\ldots,R_m]$,
there is $\Phi' \in \fo^k[E]$ such that the following holds.
\begin{itemize}\itemsep=0pt
\item
For every $\stra \models \Phi$, there is $G=(V,E)\models \Phi'$
such that $|V|=(m+3)|A|+8m+2$.
\item
For every $G=(V,E)\models \Phi'$, there is $\stra \models \Phi$
such that $|V|=(m+3)|A|+8m+2$.
\end{itemize}
Moreover, all models of $\Phi'$ are bipartite graphs.
\end{lem}

The rest of this section is devoted to the proof of Lemma~\ref{lem:preserve-k-var}.
We fix a sentence $\Phi \in\fo^k[R_1,\ldots,R_m]$, and
we assume that $z_1,\ldots,z_k$ are the variables used in $\Phi$.
Without loss of generality, we also assume that $m\geq 3$.
Moreover, we assume that $\Phi$ implies $\forall x \neg R(x,x)$,
for every $R\in \{R_1,\ldots,R_m\}$.
That is, in every model $\stra\models \Phi$,
every relation $R^{\stra}$ does not contain self-loop.
Note that self-loops can be represented by non self-loops, i.e.,
by adding a new binary relation $R'$ for each $R\in \{R_1,\ldots,R_m\}$
and replacing every atomic formula $R(x,y)$ with $R(x,y)\vee \exists y R'(x,y)$.
The intuition is that in every model $\stra\models \Phi$,
a self-loop $(u,u)\in R^{\stra}$ is represented by $(u,v)\in {R'}^{\stra}$ for some $v\neq u$.

We will first describe the main idea of our proof.
The details will be presented immediately after.
Let $\cC$ be the graph depicted in Figure~\ref{fig:C}.
It has $8m+2$ vertices, denoted by $u_1,\ldots,u_{4m+1}$ and $w_1,\ldots,w_{4m+1}$,
with the $u_i$'s being those on the left hand side, and the $w_i$'s being those on the right hand side.
The edges are $(u_i,w_i)$, for each $i=1,\ldots,4m+1$, and
$(u_i,u_{i+1})$, for each $i=1,\ldots,4m$.
Throughout this paper, we will always write $U$ and $W$ to denote
the sets $\{u_1,\ldots,u_{4m+1}\}$ and $W=\{w_1,\ldots,w_{4m+1}\}$, respectively.

\begin{figure}

\begin{picture}(360,150)(-300,-50)

\put(-230,20){$\cC:$}

\put(-110,70){\circle*{3}}
\put(-122,68){\footnotesize $u_1$}
\put(-110,67.5){\line(0,-1){15}}
\put(-107.5,70){\line(1,0){95}}

\put(-110,50){\circle*{3}}
\put(-122,48){\footnotesize $u_2$}
\put(-107.5,50){\line(1,0){95}}

\put(-10,70){\circle*{3}}
\put(-5,68){\footnotesize $w_1$}

\put(-10,50){\circle*{3}}
\put(-5,48){\footnotesize $w_2$}

\multiput(-110,45)(0,-5){10}{\circle*{0.3}}
\multiput(-10,45)(0,-5){10}{\circle*{0.3}}

\put(-110,-5){\circle*{3}}
\put(-129,-8){\footnotesize $u_{4m}$}
\put(-110,-7.5){\line(0,-1){15}}
\put(-107.5,-5){\line(1,0){95}}

\put(-110,-25){\circle*{3}}
\put(-140,-27){\footnotesize $u_{4m+1}$}
\put(-107.5,-25){\line(1,0){95}}

\put(-10,-5){\circle*{3}}
\put(-5,-7){\footnotesize $w_{4m}$}

\put(-10,-25){\circle*{3}}
\put(-5,-27){\footnotesize $w_{4m+1}$}

\end{picture}

\caption{The graph $\cC$ with $8m+2$ vertices and $8m+1$ edges.}
\label{fig:C}
\end{figure}

Let $\cD$ be the graph depicted in Figure~\ref{fig:D}.
It has $m+3$ vertices and $m+2$ edges.
The vertices are denoted by $d^P, d^Q, d^S, d^{R_1},\ldots,d^{R_m}$,
where $d^P$ is adjacent to all of $d^Q,d^{R_1},\ldots,d^{R_m}$
and $d^Q$ is adjacent to $d^S$.

\begin{figure}
\begin{picture}(360,110)(-70,-5)

\put(10,45){$\cD:$}

\put(120,70){\circle*{3}}
\put(104,68){\footnotesize $d^{P}$}
\put(120,67.5){\line(0,-1){20}}

\put(120,45){\circle*{3}}
\put(104,43){\footnotesize $d^{Q}$}
\put(120,42.5){\line(0,-1){20}}

\put(120,20){\circle*{3}}
\put(104,18){\footnotesize $d^{S}$}

\qbezier(122.5,70)(170,75)(217.5,80)
\put(220,80){\circle*{3}}
\put(225,78){\footnotesize $d^{R_1}$}

\qbezier(122.5,70)(170,65)(217.5,60)
\put(220,60){\circle*{3}}
\put(225,58){\footnotesize $d^{R_2}$}

\multiput(220,50)(0,-5){6}{\circle*{0.3}}

\qbezier(122.5,70)(170,45)(217.5,20)
\put(220,18){\circle*{3}}
\put(225,16){\footnotesize $d^{R_m}$}

\end{picture}

\caption{The graph $\cD$ with $m+3$ vertices and $m+2$ edges.}
\label{fig:D}
\end{figure}

Our intention is to construct $\Phi'$ such that 
every model $\stra \models \Phi$ with $A=\{1,\ldots,n\}$
is represented by a graph $G=(V,E) \models \Phi'$, where
there is a partition $V=V_0\cup V_1\cup \cdots \cup V_n$ and the following holds.
\begin{itemize}\itemsep=0pt
\item 
$G[V_0]$ is isomorphic to $\cC$.
\item
$G[V_i]$ is isomorphic to $\cD$, for each $i=1,\ldots,n$.
\end{itemize}
Intuitively, each element $i \in A$ is represented by $G[V_i]$.
For simplicity, we will assume that $G[V_0]$ is $\cC$ itself, i.e., $V_0=U\cup W$.
We also denote the vertices in $V_i$ by $i^{P},i^Q, i^S, i^{R_1},\ldots,i^{R_m}$
which correspond respectively to vertices $d^P, d^Q, d^S, d^{R_1},\ldots,d^{R_m}$ in $\cD$.
Each tuple $(i,j)\in R_l^{\stra}$ will then be represented by the edge $(i^{R_l},j^S)$ in $G$.
See Figure~\ref{fig:rep-Rel} for an illustration.

\begin{figure}[ht!]

\begin{picture}(360,90)(110,5)

\put(120,70){\circle*{3}}
\put(104,68){\footnotesize $i^{P}$}
\put(120,67.5){\line(0,-1){20}}

\put(120,45){\circle*{3}}
\put(104,43){\footnotesize $i^{Q}$}
\put(120,42.5){\line(0,-1){20}}

\put(120,20){\circle*{3}}
\put(104,18){\footnotesize $i^{S}$}

\qbezier(122.5,70)(170,75)(217.5,80)
\put(220,80){\circle*{3}}
\put(225,78){\footnotesize $i^{R_1}$}

\multiput(220,72)(0,-5){4}{\circle*{0.3}}

\qbezier(122.5,70)(170,60)(217.5,50)
\put(220,50){\circle*{3}}
\put(225,52){\footnotesize $i^{R_l}$}

\qbezier(222,49)(290,35)(358,21)

\multiput(220,42)(0,-5){4}{\circle*{0.3}}

\qbezier(122.5,70)(170,45)(217.5,20)
\put(220,18){\circle*{3}}
\put(225,16){\footnotesize $i^{R_m}$}


\put(360,70){\circle*{3}}
\put(344,68){\footnotesize $j^{P}$}
\put(360,67.5){\line(0,-1){20}}

\put(360,45){\circle*{3}}
\put(344,43){\footnotesize $j^{Q}$}
\put(360,42.5){\line(0,-1){20}}

\put(360,20){\circle*{3}}
\put(348,10){\footnotesize $j^{S}$}

\qbezier(362.5,70)(410,75)(457.5,80)
\put(460,80){\circle*{3}}
\put(465,78){\footnotesize $j^{R_1}$}

\qbezier(362.5,70)(410,65)(457.5,60)
\put(460,60){\circle*{3}}
\put(465,58){\footnotesize $j^{R_2}$}

\multiput(460,50)(0,-5){6}{\circle*{0.3}}

\qbezier(362.5,70)(410,45)(457.5,20)
\put(460,18){\circle*{3}}
\put(465,16){\footnotesize $j^{R_m}$}

\end{picture}

\caption{Representing a tuple $(i,j)\in R_l^{\stra}$ with the edge $(i^{R_l},j^S)$ in $G$.}
\label{fig:rep-Rel}
\end{figure}
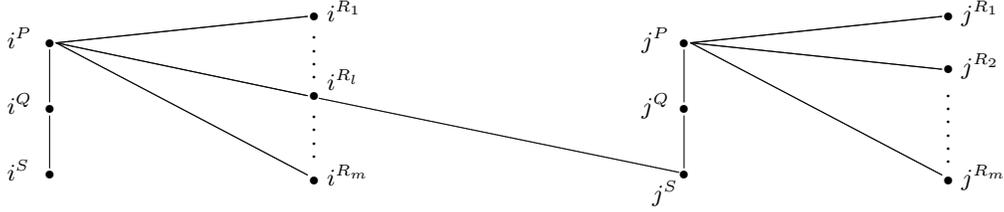

In order to achieve our intention, we differentiate 
the vertices $i^P, i^S,i^{R_1},\ldots,i^{R_m}$ by defining
them according to their connections with the vertices in $U$.
Of course, the vertices in $U$ have to be definable, as well.

We first declare the definition of the set $U$.
\begin{quote}
\begin{itemize}
\item[{($\cF_1$)}]
A vertex $u \in U$ if and only if it has degree at least 2 and exactly one of its neighbour has degree 1.
\end{itemize}
\end{quote}
The following are the properties of the set $U$ to be satisfied.
\begin{quote}
\begin{itemize}
\item[{($\cP_1$)}]
Every vertex of degree $1$ is adjacent to a vertex in $U$.
\item[{($\cP_2$)}]
There are exactly two vertices in $U$ that are adjacent to exactly one vertex in $U$.
More formally, $|X|=2$, where $X$ is the following set.
\begin{eqnarray*}
X & = & \{u \in U \mid \textrm{there is exactly one vertex} \ v \in U \ \textrm{s.t.}\ (u,v)\in E \}
\end{eqnarray*}
\item[{($\cP_3$)}]
Vertices in $U$ form a tree with diameter $\leq 4m$.
\item[{($\cP_4$)}]
Between the two vertices in the set $X$,
there is a path $\wp$ of length $4m$ that consists of only vertices in $U$.
\end{itemize}
\end{quote}
Property $\cP_1$ states that every vertex of degree $1$ is adjacent to one in $U$.
Properties $\cP_2$ and $\cP_3$ state that 
the vertices in $U$ form a tree with exactly two leaf nodes and diameter at most $4m$,
which implies that it is a line graph.
Property $\cP_4$ states that the line graph has exactly $4m+1$ vertices.

We will show that $\cF_1$ and $\cP_1$--$\cP_4$ can be defined
with first-order formulas using only three variables.
Moreover, we will also show that for every graph $G=(V,E)$ that satisfies $\cP_1$--$\cP_4$
with the set $U$ being defined as in $\cF_1$,
there is a subset $V_0\subseteq V$ such that the following holds.
\begin{itemize}\itemsep=0pt
\item
$G[V_0]$ is isomorphic to $\cC$.
\item
If a vertex $v\in V$ is either of degree $1$ or such that $v \in U$,
then $v \in V_0$.
\end{itemize}

Now, if we assume that $V_0 = U \cup W$,
and if we denote the vertices in $U$ by $u_1,\ldots,u_{4m+1}$,
we can define $u_1$ and $u_{4m+1}$ as the end vertices of the line graph $G[U]$,
whereas for each $i=2,\ldots,4m$, vertex $u_i$ is defined as the vertex with distance $i-1$ and $4m+1-i$ to $u_1$ and $u_{4m+1}$, respectively.
At this point, note that since we insist the interpretation of $E$ to be symmetric, 
our definition does not distinguish between $u_i$ and $u_{4m+2-i}$, for each $i=1,\ldots,4m+1$.

The following are the definitions of the vertices $i^P,i^Q, i^S,i^{R_1},\ldots,i^{R_m}$.
\begin{quote}
\begin{itemize}
\item[{($\cF_2$)}]
A vertex $u \in \{1^P,\ldots,n^P\}$ if and only if it is adjacent to exactly one of $u_1$ or $u_{4m+1}$,
and it is not adjacent to any other vertex in $U$.
\item[{($\cF_3$)}]
A vertex $u \in \{1^Q,\ldots,n^Q\}$ if and only if it is adjacent to exactly one of $u_2$ or $u_{4m}$,
and it is not adjacent to any other vertex in $U$.
\item[{($\cF_4$)}]
A vertex $u \in \{1^S,\ldots,n^S\}$ if and only if it is adjacent to exactly one of $u_3$ or $u_{4m-1}$,
and it is not adjacent to any other vertex in $U$.
\item[{($\cF_5$)}]
For each $R_l \in \{R_1,\ldots,R_m\}$,
a vertex $u \in \{1^{R_l},\ldots,n^{R_{l}}\}$ if and only if 
it is adjacent to exactly one of $u_{2l-1}$ or $u_{4m+1-2(l-1)}$,
and it is not adjacent to any other vertex in $U$.
\end{itemize}
\end{quote}
Again, we will show that all of them can be defined with first-order formulas
using only three variables.

Finally, to facilitate a correct representation of each relation $R_l$ 
with $\fo[E]$ formulas, we declare the following additional properties,
which can also be defined using only three variables.
\begin{quote}
\begin{itemize}
\item[{($\cP_{5}$)}]
The vertices $i^P,i^Q, i^S,i^{R_1},\ldots,i^{R_m}$
form a graph that is isomorphic to $\cD$ via the mapping
$(i^P,i^Q, i^S,i^{R_1},\ldots,i^{R_m})\mapsto (d^P,d^Q, d^S,d^{R_1},\ldots,d^{R_m})$.
\item[{($\cP_{6}$)}]
If there is an edge between the vertices in
$\{i^P,i^Q, i^S,i^{R_1},\ldots,i^{R_m}\}$ and those in $\{j^P,j^Q, j^S,j^{R_1},\ldots,j^{R_m}\}$,
where $i\neq j$,
then it is an edge between $i^{R_l}$ and $j^S$, for some $R_l \in \{R_1,\ldots,R_m\}$.
\end{itemize}
\end{quote}

With the definitions of vertices as in $\cF_1$--$\cF_5$,
we will show that for every graph $G=(V,E)$ that satisfies properties $\cP_1$--$\cP_{6}$,
there is a partition $V=V_0\cup V_1 \cup \cdots \cup V_n$ such that the following holds.
\begin{itemize}\itemsep=0pt
\item 
$G[V_0]$ is isomorphic to $\cC$.
\item
$G[V_i]$ is isomorphic to $\cD$, where $V_i =\{i^P,i^Q, i^S,i^{R_1},\ldots,i^{R_m}\}$, for each $i=1,\ldots,n$.
\item
If there is an edge between $V_i$ and $V_j$, for some $1 \leq i\neq j\leq n$,
then it is $(i^{R_l},j^S)$.
\end{itemize}
As mentioned earlier, each relation $R_l$ can then be encoded in $G$
by representing each tuple $(i,j)\in R_l^{\stra}$ with the edge $(i^{R_l},j^S)$ in $G$.

The rest of this section will be devoted to the details of the definitions of $\cF_1$--$\cF_{5}$ and $\cP_1$--$\cP_6$, 
as well as, the sentence $\Phi'$.
We divide them into five main steps.
The first step is for $\cF_1$ and $\cP_1$--$\cP_4$,
and the second step is for $\cF_2$--$\cF_5$.
The third and fourth step are for $\cP_{5}$ and $\cP_{6}$, respectively.
Finally, in the fifth step, we present the construction of the desired $\Phi'$,
where $\Phi'$ uses the same number of variables as $\Phi$.

\subsection*{Step~1: Three variable definitions for \texorpdfstring{$\cF_1$}{F\_1} and \texorpdfstring{$\cP_1$}{P\_1}--\texorpdfstring{$\cP_4$}{P\_4}.}
We will need a few auxiliary formulas.
They are all defined using three variables $x,y,z$,
which can be replaced with three arbitrary variables from among $z_1,\ldots,z_k$.

The formula $\Psi_{\deg=1}(x)$ below defines those with degree $1$.
\begin{eqnarray*}
\Psi_{\deg=1}(x) & := &
\exists y \ \Big[ E(x,y) \ \wedge \ \forall z \ \big[ E(x,z) \ \Rightarrow \ y=z\big] \Big]
\end{eqnarray*}
Next, the formula $\Psi_{U}(x)$ below defines vertices in $U$ as stated in $\cF_1$.
\begin{eqnarray*}
\Psi_{U}(x) & := &
\neg \Psi_{\deg=1}(x) \ \wedge \
\exists y \ \big[ E(x,y) \ \wedge \ \Psi_{\deg=1}(y) \big]
\end{eqnarray*}
That is, $\Psi_{U}(v)$ holds if and only if 
its degree is not 1 and it is adjacent to a vertex with degree $1$.
To avoid repetition, by abuse of terminology, when explaining the intuition of a formula,
we always write a set $U$ to mean the vertices on which $\Psi_U(x)$ holds.

We can define property $\cP_1$ with the following sentence.
\begin{eqnarray*}
\Psi_{\cP_1} & := &
\forall x \forall y\ \Big[ \big[\Psi_{\deg=1}(x) \ \wedge \ E(x,y)\big] \ \Rightarrow \ \Psi_U(y) \Big]
\end{eqnarray*}

To define the rest, we will need the following two auxiliary formulas.
\begin{itemize}\itemsep=0pt
\item
The formula $\Psi_{\fin,U}(x)$:
\begin{eqnarray*}
\Psi_{\fin,U} (x) & := & \Psi_U(x)\ \wedge\
\exists y \Big[ \Psi_U(y)\ \wedge\ E(x,y)\ \wedge\ \forall z \Big[
\big[\Psi_U(z)\ \wedge\ E(x,z)\big] \ \Rightarrow\ y=z\Big]   \Big]
\end{eqnarray*}
That is, $\Psi_{\fin,U}(v)$ holds if and only if
$v$ is in $U$ and adjacent to exactly one of the vertices in $U$.
This is intended to define the endpoints of the line graph formed by vertices in $U$.

\item
For an integer $n\geq 0$, the formula $\Psi_{U,n}(x,y)$:
\begin{eqnarray*}
\Psi_{U,0}(x,y) & := & x=y \ \wedge \ \Psi_{U}(x)
\\
\Psi_{U,n}(x,y) & := &
\Psi_{U}(x) \ \wedge \ \Psi_U(y) \ \wedge \
\exists z \ \Big[ \Psi_U(z)\ \wedge \ E(x,z) \ \wedge \ \Psi_{U,n-1}(z,y) \Big]
\end{eqnarray*}
That is, $\Psi_{U,n}(v_1,v_2)$ holds if and only if 
$\Psi_U(v_1),\Psi_U(v_2)$ hold and there is a path of length $n$
that consists of only vertices in $U$.
\end{itemize}
Now, the sentences $\Psi_{\cP_2}$, $\Psi_{\cP_3}$ and $\Psi_{\cP_4}$ that
define $\cP_2$, $\cP_3$ and $\cP_4$, respectively, are as follows.
\begin{eqnarray*}
\Psi_{\cP_2} & := & 
\exists x \exists y\ \Big[ \Psi_{\fin,U}(x) \ \wedge \ \Psi_{\fin,U}(y) \ \wedge \
\forall z \Big[\Psi_{\fin,U}(z)\ \Rightarrow\  \big[ z=x \ \vee \ z=y\big]\Big]\Big]
\\
\Psi_{\cP_3} & := & 
\forall x \forall y \
\Big[
\big[ \Psi_{U}(x) \ \wedge \ \Psi_{U}(y)\big] \ \Rightarrow \
\bigvee_{n=1}^{4m} \Psi_{U,n}(x,y)\Big] \ \wedge
\\
& &
\bigwedge_{n=1}^{4m}
\
\forall x \forall y\ 
\Big[
\big[
 \Psi_{U}(x) \ \wedge \ \Psi_{U}(y) \ \wedge \ \Psi_{U,n}(x,y)\big] \ \Rightarrow\ \bigwedge_{l\neq n \ \textrm{and} \ 1\leq l \leq 4m}
 \neg \Psi_{U,l}(x,y)\Big]
\\
\Psi_{\cP_4} & := & 
\forall x \forall y 
\Big[
\big[ \Psi_{\fin,U}(x) \ \wedge \ \Psi_{\fin,U}(y)\big] \ \Rightarrow \ \Psi_{U,4m}(x,y)\Big]
\end{eqnarray*}
Intuitively, the first line of $\Psi_{\cP_3}$ states that the vertices in $U$ form a graph with diameter $\leq 4m$,
while the second line states that the distance between two vertices in $U$ is unique.
Thus, $\Psi_{\cP_3}$ states that vertices in $U$ form a tree with diameter $\leq 4m$.
The sentence $\Psi_{\cP_4}$ states that distance between the two leaf nodes is $4m$.
Now, $\Psi_{\cP_2}$ states that there are only two leaf nodes.
So, altogether $\Psi_{\cP_2}\wedge \Psi_{\cP_3} \wedge \Psi_{\cP_4}$ states
that the set $U$ forms a line graph of $4m+1$ vertices.
Combining all these with $\Psi_{\cP_1}$,
we obtain that 
every model of $\Psi_{\cP_1}\wedge \cdots \wedge \Psi_{\cP_4}$
contains a subgraph isomorphic to $\cC$,
as stated formally below.

\begin{lem}
\label{lem:unique-C}
For every graph $G=(V,E)$, the following are equivalent.
\begin{enumerate}
\item[(a)]
$G\models \Psi_{\cP_1}\wedge \Psi_{\cP_2} \wedge \Psi_{\cP_3} \wedge \Psi_{\cP_4}$.
\item[(b)]
There is a subset $V' \subseteq V$ such that $G[V']$ is isomorphic to $\cC$.
Moreover, if a vertex $v\in V$ is either of degree $1$ or such that $\Psi_U(v)$ holds,
then $v \in V'$.
\end{enumerate}
\end{lem}
\begin{proof}
The  direction that (b) implies (a) is straightforward.
So we prove that (a) implies (b).
Assume that $G=(V,E) \ \models  \Psi_{\cP_1}\ \wedge\ \Psi_{\cP_2}\ \wedge\ \Psi_{\cP_3} \ \wedge \ \Psi_{\cP_4}$.

Let $U'$ be the set $\{u\in V \mid \Psi_U(u) \ \textrm{holds in} \ G\}$.
The sentence $\Psi_{\cP_3}$ implies that $G[U']$ is a tree of diameter $\leq 4m$, whereas
the sentence $\Psi_{\cP_2}$ implies that $G[U']$ has only two leaf nodes.
So, altogether, they imply that $G[U']$ is a line graph of at most $4m+1$ vertices.
The sentence $\Psi_{\cP_4}$ implies that it is a line graph with exactly $4m+1$ vertices.

Next, let $W'$ be the set $\{w \in V \mid \deg(w)=1\}$.
Thus, if we pick $V'=U'\cup W'$,
it follows immediately that $G[V']$ is isomorphic to $\cC$.
By $\Psi_{\cP_1}$, it is trivial that if $v\in V$ is such that either $\deg(v)=1$ or that $\Psi_U(v)$ holds,
then $v \in V'$.
\end{proof}

\subsection*{Step~2: Three variable definitions for \texorpdfstring{$\cF_2$}{F\_2}--\texorpdfstring{$\cF_5$}{F\_5}.}

The formulas $\Psi_{P}(x)$, $\Psi_{Q}(x)$, $\Psi_{S}(x)$ and $\Psi_{R_l}(x)$, for each $R_l \in \{R_1,\ldots,R_m\}$,
 below defines the vertices $i^P$'s, $i^Q$'s, $i^S$'s and $i^{R_l}$'s, respectively,
 as stated in $\cF_2$--$\cF_5$.
\begin{eqnarray*}
\Psi_{P}(x) &:= &
\exists y \Big[ 
\Psi_{\fin,U}(y) \ \wedge\
\forall z \Big[ \Psi_U(z)\ \Rightarrow \ \big[ E(x,z)\ \iff\ y=z\big]\Big]
\Big]
\\
\Psi_{Q}(x) &:= &
\exists y \Big[ 
\Psi_{\fin,U}(y) \ \wedge\
\forall z \Big[ \Psi_U(z)\ \Rightarrow \ \big[ E(x,z)\ \iff\ \Psi_{U,1}(y,z)\big]\Big]
\Big]
\\
\Psi_{S}(x) &:= &
\exists y \Big[ 
\Psi_{\fin,U}(y) \ \wedge\
\forall z \Big[ \Psi_U(z)\ \Rightarrow \ \big[ E(x,z)\ \iff\ \Psi_{U,2}(y,z)\big]\Big]
\Big]
\\
\Psi_{R_l}(x) &:= &
\exists y \Big[ 
\Psi_{\fin,U}(y) \ \wedge\
\forall z \Big[ \Psi_U(z)\ \Rightarrow \ \big[ E(x,z)\ \iff\ \Psi_{U,2l-1}(y,z)\big]\Big]
\Big]
\end{eqnarray*}

\subsection*{Step~3: Three variable definition for \texorpdfstring{$\cP_{5}$}{P\_5}.}

Intuitively, the sentence $\Psi_{\cP_{5}}$ that defines $\cP_{5}$ states the following: 
{\em For every vertex $x$ such that $\Psi_P(x)$ holds, there are vertices
$y,z,s_1,\ldots,s_m$ such that the following is true.}
\begin{itemize}\itemsep=0pt
\item
{\em $x,y,z,s_1,\ldots,s_m$ form a graph isomorphic to $\cD$.}
\item
{\em $\Psi_Q(y)$, $\Psi_S(z),\Psi_{R_1}(s_1),\ldots,\Psi_{R_m}(s_m)$ all hold.}
\end{itemize}
Such sentence can be trivially written using $m+3$ variables.
However, since each of the vertices $x,y,z,s_1,\ldots,s_m$ have distinguished definitions
and the distance between them are all bounded by a fixed length,
three variables are sufficient.

Before we proceed to the details, we need the following auxiliary formula.
For every $\alpha,\beta \in \{P,Q,S,R_1,\ldots,R_m\}$,
we define the following formula:
\begin{eqnarray*}
\Psi_{\alpha,\beta}(x,y) & := &
\Psi_{\alpha}(x) \ \wedge \ \Psi_{\beta}(y) \ \wedge \
\exists z \Big[ \Psi_{\gamma}(z) \ \wedge \ E(x,z) \ \wedge \ E(z,y)\Big],
\end{eqnarray*}
where $\gamma$ is defined according to $\alpha$ and $\beta$ as follows.
\begin{itemize}\itemsep=0pt
\item
$\gamma=Q$, when either $(\alpha,\beta)=(P,S)$ or $(\alpha,\beta)=(S,P)$.
\item
$\gamma=P$, when either $(\alpha,\beta)=(R_l,Q)$ or $(\alpha,\beta)=(Q,R_l)$, for some $R_l \in \{R_1,\ldots,R_m\}$.
\item
$\gamma=P$, for every $\alpha,\beta \in \{R_1,\ldots,R_m\}$ and $\alpha\neq \beta$.
\end{itemize}
We let $\gamma$ undefined for all the other combinations of $\alpha$ and $\beta$.
Intuitively, $\Psi_{\alpha,\beta}(x,y)$ indicates that $x$ and $y$
are the vertices in $\cD$ where $\Psi_{\alpha}$ and $\Psi_{\beta}$ hold, respectively,
and that $\Psi_{\gamma}$ holds in their middle vertex.

Now, the sentence $\Psi_{\cP_5}$ is the conjunction of the following sentences,
which for readability, are written in plain English.
\begin{itemize}\itemsep=0pt
\item
For every vertex $x$ such that $\Psi_P(x)$ holds, the following is true.
\begin{itemize}\itemsep=0pt
\item
$x$ is adjacent to exactly one vertex $y$ where $\Psi_Q(y)$ holds.
\item
For every $R_l \in \{R_1,\ldots,R_m\}$,
$x$ is adjacent to exactly one vertex $y$ where $\Psi_{R_l}(y)$ holds.
\item
There is exactly one vertex $y$ such that $\Psi_{P,S}(x,y)$ holds
and moreover, $E(x,y)$ does not hold.
\item
For every $R_l \in \{R_1,\ldots,R_m\}$,
if $y$ and $z$ are vertices adjacent to $x$ such that
$$
\Psi_Q(y) \quad \textrm{and}\quad \Psi_{R_l}(z)\quad \textrm{hold},
$$
then $E(y,z)$ does not hold.
\item
For every $R_l,R_{l'}\in \{R_1,\ldots,R_m\}$,
if $y$ and $z$ are vertices adjacent to $x$ such that
$$
\Psi_{R_{l}}(y) \quad \textrm{and}\quad \Psi_{R_{l'}}(z)\quad \textrm{hold},
$$
then $E(y,z)$ does not hold.
\end{itemize}

\item
For every vertex $x$ such that $\Psi_Q(x)$ holds, the following is true.
\begin{itemize}\itemsep=0pt
\item
$x$ is adjacent to exactly one vertex $y$ where $\Psi_P(y)$ holds.

\item
$x$ is adjacent to exactly one vertex $y$ where $\Psi_S(y)$ holds.

\item
For every $R_l \in \{R_1,\ldots,R_m\}$,
there is exactly one vertex $y$ such that $\Psi_{Q,R_l}(x,y)$ holds
and moreover, $E(x,y)$ does not hold.
\end{itemize}

\item
For every vertex $x$ such that $\Psi_S(x)$ holds, the following is true.
\begin{itemize}\itemsep=0pt
\item
$x$ is adjacent to exactly one vertex $y$ where $\Psi_Q(y)$ holds.

\item
There is exacly one vertex $y$ such that $\Psi_{S,P}(x,y)$ holds,
and moreover, $E(x,y)$ does not hold.

\item
If $y$ and $z$ are vertices such that
$$
\Psi_Q(y),\ E(x,y), \ \textrm{and}\ \Psi_{Q,R_l}(y,z) \ \textrm{hold, for some} \ R_l \in \{R_1,\ldots,R_m\},
$$
then $E(x,z)$ does not hold.
\end{itemize}

\item
For every $R_l \in \{R_1,\ldots,R_m\}$,
for every vertex $x$ such that $\Psi_{R_l}(x)$ holds, the following is true.
\begin{itemize}\itemsep=0pt
\item
$x$ is adjacent to exactly one vertex $y$ where $\Psi_P(y)$ holds.
\item
There is exactly one vertex $y$ such that $\Psi_{R_l,Q}(x,y)$ holds
and moreover, $E(x,y)$ does not hold.
\item
If $y$ and $z$ are vertices such that
$$
\Psi_P(y),\ E(x,y), \ \textrm{and}\ \Psi_{P,S}(y,z) \ \textrm{hold},
$$
then $E(x,z)$ does not hold.
\end{itemize}
\end{itemize}

Now, consider the following sentence.
\begin{eqnarray*}
\Psi_0 & := & \Psi_{\cP_1} \ \wedge \ \Psi_{\cP_2} \ \wedge \ \Psi_{\cP_3}\ \wedge \ \Psi_{\cP_4} \ \wedge \ \Psi_{\cP_5}
\ \wedge 
\\
& & 
\forall x \Big[\neg E(x,x)\Big] \ \wedge \
\forall x \Big[ \Psi_{\deg=1}(x) \ \vee \ \bigvee_{\alpha \in \{U,P,Q,S,R_1,\ldots,R_m\}} \Psi_{\alpha}(x) 
\Big].
\end{eqnarray*}
We have the following lemma.
\begin{lem}
\label{lem:partition}
For every graph $G=(V,E)\models \Psi_{0}$,
there is a partition $V=V_0\cup V_1\cup \cdots \cup V_n$ such that the following holds.
\begin{itemize}\itemsep=0pt
\item
$G[V_0]$ is isomorphic to $\cC$.
\item
For each $i=1,\ldots,n$, $G[V_i]$ is isomorphic to $\cD$, and
for every $\alpha \in \{P,Q,S,R_1,\ldots,R_m\}$,
there is exactly one node $v \in V_i$ such that $\Psi_{\alpha}(v)$ holds.
\end{itemize}
\end{lem}
\begin{proof}
Let $G=(V,E)\models \Psi_0$.
Obviously, it does not contain any self-loop.
By Lemma~\ref{lem:unique-C}, there is $V_0$ such that $G[V_0]$ is isomorphic to $\cC$.
Let $K = \{v \in V\mid \Psi_{P}(v)\ \textrm{holds}\}$.
By $\Psi_{\cP_{5}}$, for every $v\in K$, there is a set of vertices $V_v = \{u_1^v,\ldots,u_{m+2}^v\}$
such that the following holds.
\begin{itemize}\itemsep=0pt
\item
$\Psi_{R_1}(u_1^v), \ldots,\Psi_{R_m}(u_m^v), \Psi_Q(u_{m+1}^v), \Psi_S(u_{m+2}^v)$ hold. 
\item
$G[\{v\}\cup V_v]$ is isomorphic to $\cD$.
\end{itemize}
Suppose $K = \{v_1,\ldots,v_n\}$.
By $\Psi_{\cP_{5}}$ again, we have that $\{v_i\}\cup V_{v_i}$ and $\{v_j\}\cup V_{v_j}$ are disjoint, whenever $v_i \neq v_j$.

Now, for every vertex $v\in V$, either $\deg(v)=1$ or
there is a $\alpha \in \{U,P,Q,S,R_1,\ldots,R_m\}$
such that $\Psi_{\alpha}(v)$ holds.
Moreover, it is not possible that 
$\Psi_{\alpha}(v)$ and $\Psi_{\beta}(v)$ hold, for different $\alpha,\beta \in  \{U,P,Q,S,R_1,\ldots,R_m\}$.
By Lemma~\ref{lem:unique-C}, if $v$ is of degree $1$ or that $\Psi_U(v)$ holds, then $u\in V_0$.
Otherwise, $v \in V_i$, for some $i=1,\ldots,n$.
Thus, $V$ is partitioned into $V_0\cup V_1\cup \cdots \cup V_n$.
This completes our proof.
\end{proof}

\subsection*{Step~4: Three variable definition for \texorpdfstring{$\cP_6$}{P\_6}.}

Before we define the sentence  for $\cP_{6}$,
we need the following terminology.
Let $G=(V,E)\models \Psi_0$.
We say that two vertices {\em $u,v\in V$ are in the same $\cD$-component},
if there is $V'\subseteq V$ such that the following holds.
\begin{itemize}\itemsep=0pt
\item
$u,v \in V'$.
\item
$G[V']$ is isomorphic to $\cD$.
\item
For every $\alpha \in \{P,Q,S,R_1,\ldots,R_m\}$,
there is exactly one $w \in V'$ such that $\Psi_{\alpha}(w)$ holds.
\end{itemize}
We can define a three-variable formula $\Psi_{\comp}(x,y)$ such
that $\Psi_{\comp}(x,y)$ holds if and only if $x$ and $y$ are in the same $\cD$-component.
This can be done as follows.
Suppose $\alpha = S$ and $\beta=R_l$,
and that $\Psi_{\alpha}(x)$ and $\Psi_{\beta}(y)$ hold.
Then, $x$ and $y$ are in the same $\cD$-component
is equivalent to stating that {\em there is $z$ such that $E(x,z)$, $\Psi_Q(z)$ and $\Psi_{Q,R_l}(z,y)$ hold.}
We can enumerate similar formulas for every possible $\alpha$ and $\beta$,
and conjunct them all to obtain a formula $\Psi_{\comp}(x,y)$
that asserts whether $x$ and $y$ are in the same $\cD$-component.

Now, the sentence $\Psi_{\cP_{6}}$ that defines $\cP_{6}$
states as follows.
For every adjacent vertices $x$ and $y$,
if they are not in the same $\cD$-component,
then for some $R_l \in \{R_1,\ldots,R_m\}$,
either one of the following holds.
\begin{itemize}\itemsep=0pt
\item
$\Psi_S(x)$ and $\Psi_{R_l}(y)$ hold.
\item
$\Psi_{R_l}(x)$ and $\Psi_{S}(y)$ hold.
\end{itemize}

The following lemma is immediate from Lemma~\ref{lem:partition}
and the intended meaning of $\Psi_{\cP_{6}}$.

\begin{lem}
\label{lem:representation}
For every graph $G=(V,E)$,
if $G\models \Psi_0 \wedge \Psi_{\cP_{6}}$,
then $V$ can be partitioned into $V=V_0\cup V_1 \cup \ldots\cup V_n$ such that
the following holds.
\begin{itemize}\itemsep=0pt
\item 
$G[V_0]$ is isomorphic to $\cC$.
\item
For each $i=1,\ldots,n$, $G[V_i]$ is isomorphic to $\cD$, and
for every $\alpha \in \{P,Q,S,R_1,\ldots,R_m\}$,
there is exactly one node $v \in V_i$ such that $\Psi_{\alpha}(v)$ holds.
\item
If there is an edge $(u,v)$ such that $u \in V_i$ and $v\in V_j$, for some $1\leq i \neq j \leq n$,
then either $\Psi_S(v), \Psi_{R_l}(u)$ hold or $\Psi_S(u), \Psi_{R_l}(v)$ hold, for some $R_l \in \{R_1,\ldots,R_m\}$.
\end{itemize}
\end{lem}

Note also that every graph $G=(V,E)$ that satisfies $\Psi_0\wedge \Psi_{\cP_{6}}$ is indeed a bipartite graph.
Using the same notation as in Lemma~\ref{lem:partition},
we  assume that $G[V_0]$ is $\cC$ itself.
Furthermore, we also denote by $V_i = \{i^P,i^Q,i^S,i^{R_1},\ldots,i^{R_l}\}$,
where the mapping $(i^P,i^Q,i^S,i^{R_1},\ldots,i^{R_m})\mapsto (d^P,d^Q,d^S,d^{R_1},\ldots,d^{R_m})$
is an isomorphism from $G[V_i]$ to $\cD$.
Then, $G$ is a bipartite graph with the partition $V= V'\cup V''$, where
\begin{eqnarray}
\label{eq:Va}
V' & = & \{u_1,u_3,\ldots,u_{4m+1}\}\cup \{w_2,w_4,\ldots,w_{4m}\} \cup \{i^Q,i^{R_1},\ldots,i^{R_l}\mid i = 1,\ldots,n\}
\\
\label{eq:Vb}
V'' & = & \{u_2,u_4,\ldots,u_{4m}\}\cup \{w_1,w_3,\ldots,w_{4m+1}\} \cup \{i^P,i^S\mid i = 1,\ldots,n\}
\end{eqnarray}

\subsection*{Step~5: The construction of \texorpdfstring{$\Phi'$}{Phi'}.}

First, for each formula $\varphi(\vz)$ of $\Phi$, where $\vz=(z_1,\ldots,z_t)$ and $t\geq 3$,
we construct $\widetilde{\varphi}(\vz)$ with the same free variables $\vz$
inductively as follows.
\begin{description}
\item[Base case]
$\varphi(\vz)$ is an atomic formula $R_l(x,y)$, i.e., $\vz=(x,y)$ and $x,y \in \{z_1,\ldots,z_t\}$.
Then,
\begin{eqnarray*}
\widetilde{\varphi}(x,y) & := & \Psi_P(x) \ \wedge \ \Psi_P(y) \ \wedge
\\
& &
\exists z 
\arraycolsep=1.4pt\def\arraystretch{1.6}
\left[
\begin{array}{l}
\Psi_Q(z)\  \wedge \ E(y,z)\  \wedge 
\\
\exists y
\left[
\begin{array}{l}
\Psi_S(y) \ \wedge \
E(z,y)\  \wedge 
\
\exists z 
\left[
\begin{array}{l}
\Psi_{R_l}(z)\ \wedge
\
E(x,z)\   \wedge\ E(y,z)  
\end{array}
\right]
\end{array}
\right]
\end{array}
\right]
\end{eqnarray*}
The variable $z$ is such that $z\in \{z_1,\ldots,z_t\}$ and $z\neq x,y$.
Note also that variables $y$ and $z$ are being reused.

The intuitive meaning of $\widetilde{\varphi}(x,y)$ is as follows.
Assuming that $\Psi_P(x)$ and $\Psi_P(y)$ hold,
$\widetilde{\varphi}(x,y)$ states that there are three vertices $v,v',v''$ such that
the following holds.
\begin{itemize}\itemsep=0pt
\item
$\Psi_Q(v)$, $\Psi_S(v')$, $\Psi_{R_l}(v'')$ hold.
\item
$(y,v),(v,v'),(x,v'')$ and $(v'',v')$ are edges.
\end{itemize}

In a similar way, when $\varphi(\vz)$ is an atomic formula $x=y$, then,
\begin{eqnarray*}
\widetilde{\varphi}(x,y) & := & \Psi_P(x) \ \wedge \ \Psi_P(y) \ \wedge \ x=y.
\end{eqnarray*}

\item[Induction step]
\begin{eqnarray*}
\widetilde{\varphi}(\vz) & := &
\arraycolsep=1.4pt\def\arraystretch{1.2}
\left\{
\begin{array}{ll}
\widetilde{\varphi}_1(\vz) \wedge \widetilde{\varphi}_2(\vz),
&\qquad \mbox{if} \ \varphi(\vz)\ \mbox{is} \ \varphi_1(\vz)\wedge \varphi_2(\vz)
\\
\neg \widetilde{\varphi}_1(\vz),
&\qquad \mbox{if} \ \varphi(\vz) \ \mbox{is}\ \neg\varphi_1(\vz)
\\
\exists x \ \Psi_P(x) \wedge \widetilde{\varphi}_1(x,\vz),
&\qquad \mbox{if} \ \varphi(\vz)\ \mbox{is}\ \exists x \ \varphi_1(x,\vz)
\end{array}
\right.
\end{eqnarray*}
\end{description}
Note that $\Phi'$ uses the same number of variables as $\Phi$.

We have the following lemma which states that
$\Phi$ and $\Phi'$ are equi-satisfiable.
 
\begin{lem}
\label{lem:final}
For every formula $\varphi(z_1,\ldots,z_t)\in \fo^k[R_1,\ldots,R_m]$,
the following holds.
\begin{itemize}\itemsep=0pt
\item
For every structure $\stra = \langle A, R_1^{\stra},\ldots,R_m^{\stra}\rangle$,
for every $i_1,\ldots,i_t\in A$ such that 
\begin{eqnarray*}
\stra & \models & \varphi(i_1,\ldots,i_t),
\end{eqnarray*}
there is a graph $G=(V,E)$ and $u_1,\ldots,u_t\in V$ such that
\begin{eqnarray*}
G & \models & \Psi_0 \ \wedge \ \Psi_{\cP_6}\ \wedge \ \widetilde{\varphi}(u_1,\ldots,u_t).
\end{eqnarray*}

\item
Vice versa, 
for every graph $G=(V,E)$ and for every $u_1,\ldots,u_t\in V$ such that
\begin{eqnarray*}
G & \models & \Psi_0 \ \wedge \ \Psi_{\cP_6}\ \wedge \ \widetilde{\varphi}(u_1,\ldots,u_t),
\end{eqnarray*}
there is a structure $\stra = \langle A, R_1^{\stra},\ldots,R_m^{\stra}\rangle$
and $i_1,\ldots,i_t\in A$ such that 
\begin{eqnarray*}
\stra & \models & \varphi(i_1,\ldots,i_t).
\end{eqnarray*}
\end{itemize}
\end{lem}
\begin{proof}
For a structure $\stra = \langle A, R_1^{\stra},\ldots,R_m^{\stra}\rangle$,
where $A = \{1,\ldots,n\}$,
let $G=(V,E)$ be the following graph.
\begin{itemize}\itemsep=0pt
\item
$V= U \cup W \cup V_1 \cup  \cdots  \cup V_n$,
where each $V_i = \{i^P,i^Q,i^S,i^{R_1},\ldots,i^{R_m}\}$
and $U=\{u_1,\ldots,u_{4m+1}\}$ and $W=\{w_1,\ldots,w_{4m+1}\}$.
\item
$G[U \cup W]$ is isomorphic to $\cC$
and $G[V_i]$ is isomorphic to $\cD$, for each $i=1,\ldots,n$.
\item
Every vertex $u \in \{1^P,\ldots,n^P\}$ is adjacent to $u_1$,
and not to any other vertex in $U$.
\item
Every vertex $u \in \{1^Q,\ldots,n^Q\}$ is adjacent to $u_2$,
and not to any other vertex in $U$.
\item
Every vertex $u \in \{1^S,\ldots,n^S\}$ is adjacent to $u_3$,
and not to any other vertex in $U$.
\item
For each $R_l \in \{R_1,\ldots,R_m\}$,
every vertex $u \in \{1^{R_l},\ldots,n^{R_{l}}\}$ is adjacent to $u_{2l-1}$,
and not adjacent to any other vertex in $U$.
\item
For each $R_l \in \{R_1,\ldots,R_m\}$,
for each $(i,j)\in R_l^{\stra}$,
we have an edge $(i^{R_l},j^S)$ in $G$.
\end{itemize}
By straightforward induction on formula $\varphi(z_1,\ldots,z_t)$,
we can establish the following. For every $i_1,\ldots,i_t \in A$:
\begin{eqnarray*}
\stra\ \models\ \varphi(i_1,\ldots,i_t)
& \mbox{if and only if} &
G\ \models\ \Psi_0\wedge\Psi_{\cP_{6}}\wedge \widetilde{\varphi}(i_1^P,\ldots,i_t^P).
\end{eqnarray*}

Vice versa,
let $G=(V,E)\models \Psi_0 \wedge \Psi_{\cP_{6}}$.
Let $V_0\cup V_1\cup \cdots\cup V_n$ be the partition of $V$,
where $V_i=\{i^P,i^Q,i^S,i^{R_1},\ldots,i^{R_m}\}$, for each $i=1,\ldots,n$,
as in Lemma~\ref{lem:representation}.
We can define a structure $\stra = \langle A, R_1^{\stra},\ldots,R_m^{\stra}\rangle$
as follows.
\begin{itemize}\itemsep=0pt
\item
$A=\{1,\ldots,n\}$.
\item
For each $R_l \in \{R_1,\ldots,R_m\}$,
for every edge $(i^{R_l},j^S)$ in $G$,
we have $(i,j)\in R_l^{\stra}$.
\end{itemize}
Again, by straightforward induction on formula $\varphi(z_1,\ldots,z_t)$,
we can establish the following. For every $i_1,\ldots,i_t \in A$:
\begin{eqnarray*}
\stra\models \varphi(i_1,\ldots,i_t)
& \mbox{if and only if} &
G\models \Psi_0\wedge\Psi_{\cP_{6}}\wedge \widetilde{\varphi}(i_1^P,\ldots,i_t^P).
\end{eqnarray*}
This completes our proof.
\end{proof}

To complete our proof of Lemma~\ref{lem:preserve-k-var},
we set $\Phi'$ as follows.
\begin{eqnarray*}
\Phi' & := & \Psi_0 \ \wedge \ \Psi_{\cP_{6}} \ \wedge \ \widetilde{\Psi}
\end{eqnarray*}
That $\Phi'$ is the desired sentence follows immediately from Lemmas~\ref{lem:representation} and~\ref{lem:final}.

Note also that for $G\models \Psi_0\wedge\Psi_{\cP_{6}}\wedge \widetilde{\Phi}$,
the additional edge needed to represent the relation $R_l^{\stra}(i,j)$ in $G$ is between $i^{R_l}$ and $j^S$,
thus the partition $V'\cup V''$ as defined in Equations~(\ref{eq:Va}) and~(\ref{eq:Vb}) still preserves
the bipartite-ness of $G$.

\section{Concluding remarks}
\label{sec:concl}

In this paper we have shown that the spectrum of a sentence 
using at least three variables and binary relation symbols
is linearly proportional to the spectrum of a sentence using the same amount of variables
and only one symmetric binary relation symbol $E$, whose models are all bipartite graphs (Theorem~\ref{theo:main}).
Building from our previous work~\cite[Corollary~3.5]{KT15tocl},
we obtain that to settle Asser's conjecture,
it is sufficient to consider only sentences using only three variables on bipartite graphs (Corollary~\ref{cor:asser}),
i.e., the following two sentences are equivalent.
\begin{itemize}\itemsep=0pt
\item 
The class of first-order spectra is closed under complement.
\item
For every three-variable sentence $\varphi$ 
whose models are all undirected bipartite graphs,
the complement of $\spec(\varphi)$ is also a spectrum.
\end{itemize}

The proof of Corollary~\ref{cor:asser} follows closely the one in~\cite[Corollary~3.5]{KT15tocl}.
The direction from the first to the second bullet is trivial.
The other direction is as follows.
Define the following class $\cC$.
\begin{eqnarray*}
\cC & := &
\left\{
\begin{array}{l|l}
\spec(\phi) &
\begin{array}{l}
\phi \ \mbox{uses only three variables}
\\ 
\mbox{and its models are all undirected bipartite graphs}
\end{array}
\end{array} 
\right\}
\\
\cocC & := &
\left\{
\begin{array}{l|l}
\bbN- S &
S \in \cC
\end{array} 
\right\}
\end{eqnarray*}
Suppose that the second bullet holds, i.e., $\cocC \subseteq \spec$.
Let $S$ be a set of integers such that $A\in \ntime{2^n}$,
where the input number is written in binary form.
In~\cite{KT15tocl}, we have already shown that $S$ is the spectrum
of a three-variable sentence using only binary relations.
By Theorem~\ref{theo:main},
there is $p$ and $q$ such that
the set $S' =\{px+q \mid x \in S\} \in \cC$.
By the assumption that $\cocC \subseteq \spec$,
we have that $\bbN-S' \in \spec=\ne$.
Since addition/subtraction/multiplication/division by constant
can be performed in linear time,
we have $\bbN-S \in \ne$.
By padding argument,
this implies that for every set $S \in \ne$,
the complement $\bbN - S$ also belongs to $\ne$.
Then, Corollary~\ref{cor:asser} follows immediately from $\ne=\spec$.

Note that Corollary~\ref{cor:asser} reduces Asser's conjecture in two directions:
First, it reduces the number of variables to three, and
second, it reduces to sentences whose models are all undirected bipartite graphs.
It should be remarked that bipartite-ness is not first-order definable,
thus, it will be interesting to obtain a characterization of sentences whose models are all bipartite graphs.
We leave this as future work.

It will also be interesting to show whether
the linear blowup in Theorem~\ref{theo:main} is necessary.
As pointed out in the introduction,
Durand, et. al. showed that there is a strong evidence that 
collapsing the class of spectra involving arbitrary number of unary functions
to a fixed number of unary functions is likely to be difficult~\cite[Proposition~5.1]{DurandFL97}.
Similar evidence for Theorem~\ref{theo:main} will be interesting.

\subsection*{Acknowledgement}
The authors would like to thank the anonymous referees for their excellent comments.
They are also grateful to Arnaud Durand and Etienne Grandjean
for their helpful comments that greatly improve the earlier version of this paper.
The proof in Remark~\ref{remark}, which is simpler than our original proof,
is due to them.
The second author acknowledges the generous financial support of 
Taiwan Ministry of Science and Technology under the grant no.~105-2221-E-002-145-MY2.

\bibliographystyle{plainurl}

\end{document}